\documentclass[conference]{IEEEtran}
\IEEEoverridecommandlockouts

\usepackage{amsmath,amssymb,amsfonts}
\usepackage{algorithmic}
\usepackage{graphicx}
\usepackage{textcomp}
\usepackage{xcolor}

\def\BibTeX{{\rm B\kern-.05em{\sc i\kern-.025em b}\kern-.08em
    T\kern-.1667em\lower.7ex\hbox{E}\kern-.125emX}}
    
\usepackage{cite}  

\usepackage{cleveref}
\usepackage{comment}
\usepackage{algorithm}
\usepackage{booktabs}
\usepackage{enumitem}

\usepackage{times}
\usepackage{courier}
\usepackage{helvet}
\usepackage[hyphens]{url}
\usepackage{caption}
\usepackage{subcaption}
\usepackage{mathtools}
\usepackage{accents}

\usepackage{tikz}
\usetikzlibrary{calc}

\usepackage{stackengine}

\usepackage{amsthm}
\theoremstyle{plain}
\newtheorem{theorem}{Theorem} 
\newtheorem{lemma}{Lemma}

\theoremstyle{definition}
\newtheorem{definition}{Definition}
\newtheorem{example}{Example}[section]
\newtheorem{assumption}{Assumption}[section]

\theoremstyle{remark}
\newtheorem{remark}{Remark}[section]

\begin{document}
\title{
Spectral Flow Learning Theory: 
Finite-Sample Guarantees for System Identification 
}

\author{Chi Ho Leung and Philip E. Par\'{e}*
    \thanks{*Chi Ho Leung and Philip E. Par\'e are with the Elmore Family School of Electrical and Computer Engineering, Purdue University, USA.
    E-mail: leung61@purdue.edu, philpare@purdue.edu. 
    This material is based upon work supported in part by the National Science Foundation (NSF-ECCS \#2238388).
    }
}

\maketitle

\begin{abstract}
We study the identification of continuous-time vector fields from irregularly sampled trajectories. 
We introduce spectral flow learning, which learns in a windowed flow space using a lag-linear label operator that aggregates lagged Koopman actions.
We provide finite-sample, high-probability (FS-HP) guarantees for the class of variable-step linear multistep methods (vLMM).
The FS-HP rates are constructed using spectral regularization with qualification-controlled filters for flow predictors under standard source and filter assumptions.
A multistep observability inequality links flow error to vector-field error and yields two-term bounds that combine a statistical rate with an explicit discretization bias from vLMM theory.
Simulations on a controlled mass-spring system corroborate the theory and clarify conditioning, step-sample tradeoffs, and practical implications.
\end{abstract}


\section{Introduction}

Classical continuous-time identification lacks operator-theoretic flow-field bounds \cite{garnier2008direct}, which motivates the proposed \emph{spectral flow learning} (SFL) framework that offers FS-HP guarantees 
for identification pipelines for linear and nonlinear systems. 
SFL offers a spectral-regularization view of learning dynamical laws from irregular trajectories.
In particular, learning is posed as an inverse problem \cite{engl2015regularization}, driven by the inclusion map $\mathsf I$ of a vector-valued Reproducing Kernel Hilbert Space (vvRKHS) into a population $L^2$ space. 
We provide operator-theoretic FS-HP guarantees for two broad classes of system-identification pipelines---(i) flow learning and (ii) vector-field learning. 

Flow learning treats the one-step forward transition map $\Phi$ as the regression target, whereas vector-field learning targets the instantaneous dynamics $f$. 
We establish FS-HP generalization guarantees for one-step transition learners under irregular sampling, and then transfer these guarantees to vector-field estimators by (i) coupling $\Phi$ to $f$ through a backward-looking multistep window and (ii) disentangling a data-driven statistical rate from an explicit discretization bias induced by the step schedule. 
This two-layer analysis covers a broad spectrum of methods: flow-learning pipelines such as ARX, subspace identification, the DMD family, and GP-SSMs.
The vector-field pipelines covered include GP vector-field regression.

Section~\ref{sec:setup} fixes notation, specifies the control-affine dynamics and their flows, introduces the design and step laws, and states the FS-HP targets that determine the population spaces and bounded operators. 
Section~\ref{sec:variable} records the ratio-parameterized multistep form, the bounded-coefficient/step-ratio and order-$p$ consistency assumptions, and proves a uniform local-truncation estimate that controls deterministic bias~\cite[Ch.~3.5]{hairer1993solving}. 
Section~\ref{sec:spectral} assembles the vvRKHS hypotheses, inclusion/covariance/integral operators, and the spectral filter with qualification and Lipschitz properties, together with flow- and vector-field-level source conditions~\cite{de2005learning,bauer2007regularization}. 
Section~\ref{sec:sfl} defines the label map and vLMM residual via Koopman-lag forcing and state aggregation, establishes an observability inequality, and combines these ingredients with spectral regularization to obtain excess-risk bounds in flow space and their lifting to vector field, cleanly separating statistical rates from discretization bias. 
Section~\ref{sec:numeric} evaluates representative pipelines on a controlled mass-spring system, illustrating the issue of benchmarking across source conditions and the fast-sampling paradox.
Our operator notation and bias-variance decomposition follow \cite{rosasco2005spectral}.


\section{Setup and Problem Formulation}\label{sec:setup}
This section makes the problem setup precise. 
We first define notation and basic functional-analytic conventions. 
We then specify the dynamical system and its flows (autonomous and input-driven), introduce the design and step law that governs state visitation time advances on irregular grids, respectively.
Finally, we state the finite-sample, high-probability bounds to be proved in Eq.~\eqref{eq:fs-hp-target}. 
These choices determine the population spaces and bounded operators used throughout.

\subsection{Notation}
\(\mathbb{R}\) denotes the real number line. Vectors in \(\mathbb{R}^n\) are column vectors. For a set \(A\), \(\mathbf{1}_A\) denotes its indicator function. 
For a test function \(\varphi\) and a measure \(\mu\), the pairing is defined by
\(\langle\varphi,\mu\rangle \coloneqq \int \varphi d\mu.\)
We write \(a\lesssim b\) when there exists a universal constant \(C>0\) with \(a\le C b\), and we write \(a\asymp b\) when \(a\lesssim b\) and \(b\lesssim a\).
For a Hilbert space \(\mathcal H\), \(\langle u,v\rangle_{\mathcal H}\) denotes the inner product and \(\|u\|_{\mathcal H}\) denotes the norm. 
On \(\mathbb R^n\), the inner product is \(\langle x,y\rangle = x^\top y\) and the Euclidean norm is \(\|x\|_2=(x^\top x)^{1/2}\). 
The space \(L^2(\mathcal X,\mu)\) consists of square-integrable functions and its norm is
\(\|f\|_{L^2} = \sqrt{\int |f|^2 d\mu}.\)
The space \(\mathcal C(\mathcal X)\) denotes continuous real-valued functions on \(\mathcal X\), and \(\mathcal C^k(\mathcal X)\) denotes the \(k\)-times continuously differentiable functions.
The Banach space \(\mathcal L(\mathcal H_1,\mathcal H_2)\) denotes bounded linear operators \(\mathsf T:\mathcal H_1\to\mathcal H_2\), and the operator norm is
\(\|\mathsf T\|_{\mathcal L} = \sup_{\|x\|=1}\|\mathsf T x\|.\)
The adjoint of \(\mathsf T\) is written \(\mathsf T^*\), and \(\mathrm{id}\) denotes the identity operator. An operator is self-adjoint when \(\mathsf T=\mathsf T^*\), and it is positive semidefinite when \(\mathsf T\succeq 0\). 
For a self-adjoint operator \(\mathsf T\), \(\lambda_{\min}(\mathsf T)\) and \(\lambda_{\max}(\mathsf T)\) denote its minimal and maximal eigenvalues.
The set \(\mathcal P(\mathcal X)\) denotes Borel probability measures on \(\mathcal X\). 
For \(\rho\in\mathcal P(\mathcal X)\), \(\operatorname{supp}(\rho)\) denotes the support and \(\delta_x\) denotes the Dirac mass at \(x\). 
We write \(\mathbb E[\cdot]\) for expectation and \(\mathbb E[\cdot|\cdot]\) for conditional expectation.
We write \(a_n=\mathcal O(b_n)\) when \(\sup_n |a_n/b_n|<\infty\).

\subsection{Dynamical System and Flows}
In SFL, the flow is the primary regression object: we view the problem of learning vector-fields from trajectory data as identifying the one-step map $(x,h)\mapsto \Phi^{h}(x)$ in an $L^2$ space and then lift the flow estimator to the vector field. 
This choice makes the feature-label pairs well-posed under irregular sampling and enables us to apply spectral regularization methods in learning theory to system identification problems.
With this motivation, we now define the dynamical model and its flows.
Let $\mathcal X\subset\mathbb R^n$ be compact and $\mathcal B=\mathcal B(\mathcal X)$ its Borel $\sigma$-algebra. 
Consider the control‑affine dynamics:
$$
\dot x(t)=f(x(t)) + g(x(t)) u(t),\qquad x(0)=x_0,
$$
with $f:\mathcal X\to\mathbb R^n$, $g:\mathcal X\to\mathbb R^{n\times m}$ locally Lipschitz, and $u:\mathbb R_{\ge 0}\to\mathbb R^m$ measurable, locally bounded. 
Assume forward invariance: $x(t)\in\mathcal X$ for $t\ge 0$.

Furthermore, let autonomous flow with zero input $\phi^t:\mathcal X\to\mathcal X$ solve $\tfrac{d}{dt}\phi^t(x)=f(\phi^t(x))$ with $\phi^0=\mathrm{id}$ and semigroup property $\phi^{t+s}=\phi^t\circ\phi^s$.
Define input-driven flow $\Phi_u^{t,s}:\mathcal X\to\mathcal X$ that maps $x(s)\mapsto x(t)$ and satisfies:
\begin{align*}
    \tfrac{d}{dt}\Phi_u^{t,s}(x) &= f(\Phi_u^{t,s}(x)) + g(\Phi_u^{t,s}(x)) u(t),\\
    \Phi_u^{s,s}&=\mathrm{id}, \quad \Phi_u^{t,s}=\Phi_{\tau^s u}^{t-s,0}.
\end{align*}
where $\tau^s u$ is the time-shift operator of the input signal $u$ by $s$, such that $(\tau^s u)(t) \coloneqq u(t+s), \forall t\ge 0$.
We write $\Phi_u^{t}\coloneqq\Phi_u^{t,0}$. Special case: $u\equiv 0\Rightarrow \Phi_0^t=\phi^t$.
We denote flows generated by an autonomous vector field $f$ and control-affine dynamics $(f, g)$ as $\phi^h(\ \cdot\ ; f)$ and $\Phi^h_u(\ \cdot\ ; f, g)$, respectively.
Furthermore, we use $\Phi^h(\cdot)$ when the discussion is agnostic with respect to the existence of a system input $u$, and $\Phi^h_\rho(\cdot)$ when the input is determined by a design law $\rho$.

\subsection{Design Law (State Visitation Distribution)}
The \emph{design law} $\rho_{\mathcal X}^{\mathrm{design}}\in\mathcal P(\mathcal X)$, specifies where the experiment spends time in $\mathcal X$, and
$\rho_{\mathcal X}^{\mathrm{design}}$ is a Borel probability measure on $(\mathcal X,\mathcal B)$.
There are multiple ways to construct a design law; we are interested in the
continuous-time time-average law. 
Along a single path $x(t)=\Phi_u^{t}(x_0)$:
      \begin{equation*}
          \lim_{T\to \infty}\bigg\langle\varphi, \frac{1}{T}\int_0^T\delta_{x(t)}dt\bigg\rangle = \langle\varphi, \rho_{\mathcal X}\rangle, \quad \forall\ \varphi\in \mathcal{C}(\mathcal{X}).
      \end{equation*}
      where $\delta_{x}$ is the Dirac measure: $\delta_x(A)=\mathbf 1_{A}(x)$, $A\in\mathcal B$.
The design law $\rho_{\mathcal X}^{\mathrm{design}}$ determines coverage of states; identifiability over the identification domain $\mathcal X_{\rm ID}$ requires $\mathcal X_{\rm ID}\subseteq \operatorname{supp}(\rho_{\mathcal X}^{\mathrm{design}})$.
\begin{remark}
    The design law induces the geometry in which we measure errors and therefore governs identifiability.
    In particular, we can view identifiability as an input-design problem: choose an input policy $u(\cdot)$ so that the induced $\rho_{\mathcal X}^{\mathrm{design}}(u)$ places mass on the domain of interest.
    Moreover, the population risks $\mathcal{E}(\Phi_\rho)$ and $\mathcal{E}(f_\rho)$ used later in the problem formulation are squared $L^2$-errors with respect to the design weights $\rho^{\rm design}_{\mathcal X} \times \rho_{\mathbf h}$ and $\rho^{\rm design}_{\mathcal X}$, respectively.
    The requirement of $\mathcal X_{\rm ID}\subseteq \operatorname{supp}(\rho_{\mathcal X}^{\mathrm{design}})$ ensures that discrepancies on the identification domain are penalized and thus identifiable.
    When a specific identification algorithm uses inputs explicitly,
    we overload \(x\) as \(x\coloneqq(x,u)\) and regard \(\rho_{\mathcal X}^{\mathrm{design}}\)
    as the joint law of sampled states and inputs.
\end{remark}

\subsection{Step Law (Time‑Advance Distribution)}
Let $\boldsymbol{h} \coloneqq (0,h_{\max}]$;
the \emph{step law} $\rho_{\boldsymbol h}\in\mathcal{P}\big(\boldsymbol{h}\big)$ specifies how far we advance in time when forming one‑step labels $(x_k,h_k)\mapsto x_{k+1}$.
It is a probability measure on the interval $(0,h_{\max}]$.
A step law $\rho_{\boldsymbol h}$ can be constructed through index‑average:
  \begin{equation*}
      \lim_{N\to \infty}\bigg\langle\varphi, \frac{1}{N}\sum_{k=0}^{N-1}\delta_{h_k}\bigg\rangle = \langle\varphi, \rho_{\boldsymbol h}\rangle, \quad \forall\ \varphi\in \mathcal{C}(\boldsymbol{h}).
  \end{equation*}
When modeling a step-window of length $M$, we consider a step-window $[h_1,\dots,h_{M-1}]$ together with the window anchor, $h_{\rm anc}$.
Let the bounded step-ratio vector be:
$\zeta(M-1) \coloneqq [\zeta_1, \dots, \zeta_{M-1}] \in\boldsymbol\zeta\subset(0,\bar\zeta]^{M-1},$
where $\zeta_k \coloneqq h_k/h_{k-1}$ and $\zeta_{M-1} = h_{M}/h_{M-1}$.
We define a $M$-steps window as a tuple $T_M \coloneqq (h_{\rm anc}, \zeta(M-1)) \in \boldsymbol{h}\times \boldsymbol{\zeta} \eqqcolon \mathcal{T}$, with $\rho_{\mathcal{T}} \in \mathcal{P}(\mathcal{T})$.
We are ready to formulate the problem of identifying the flow and vector field FS-HP bounds in the SFL framework.

\subsection{Problem Formulation (FS-HP bound)}
Fix the design law $\rho_{\mathcal X}^{\mathrm{design}}\in\mathcal P(\mathcal X)$ and the step law
$\rho_{\boldsymbol h}\in\mathcal P((0,h_{\max}])$. 
Let $\ell$ denote the number of effective samples and let $H\sim\rho_{\boldsymbol h}$ denote the anchor step, where $\mathbb E[H^{q}]$ reduces to $h^{q}$ for a deterministic schedule.
We define the flow feature-label pairs $\{((x_i, h_i), y_i)\}_{i=1}^\ell$.
Let $\rho((x,h),y)=\rho(y\mid x,h)\rho_{\mathcal X,\boldsymbol h}^{\rm design}(x,h)$ with $\rho_{\mathcal X,\boldsymbol h}^{\rm design} \coloneqq \rho_{\mathcal X}^{\mathrm{design}} \times \rho_{\boldsymbol h}$.
We denote the population flow space as $\mathcal{W}$ and likewise the population vector-field space as $\mathcal{V}$.
The target flow $\Phi_\rho \in \mathcal{W}$ is:
\begin{align*}
    \Phi^h_\rho(x) \coloneqq \mathbb E[Y\mid X=x,H=h] =\int_{\mathcal X}y d\rho(y\mid x,h).
\end{align*}
Similarly, the target vector-field $f_\rho \in \mathcal{V}$ is defined as:
\begin{equation*}
    f_\rho(x)
    \coloneqq
    \lim_{\epsilon\downarrow 0} 
    \mathbb{E} \left[\frac{Y-X}{H}\bigg| X=x,0<H\le \epsilon\right],
\end{equation*}
Under the flow source condition
$\Phi_\rho\in\Omega^{(\mathcal W)}_{r,R}$, \cite[Eq.~9]{rosasco2005spectral}, and a spectral filter ${\mathsf g}_\lambda$ with
qualification $\nu\ge r$ and Lipschitz exponent $\mu$ (set $\beta\coloneqq\max\{1,2\mu\}$)~\cite[Def.~1]{rosasco2005spectral}, 
our goal is to establish \emph{finite-sample, high-probability} bounds of the form:
\begin{align}\label{eq:fs-hp-target}
    \underbrace{\|\widehat{\Phi}-\Phi_\rho\|_{\mathcal W}^{2}}_{\text{flow excess risk}}
    & \le 
    \log \frac{4}{\eta}  C_{\rm flow} \ell^{-\frac{2r}{2r+\beta}},\\[3pt]
    \underbrace{\|\widehat f-f_\rho\|_{\mathcal V}^{2}}_{\text{vector-field excess risk}}
    & \le 
    \frac{C_{\rm fit}}{c_{\rm obs}^{2}(h)} \log \frac{4}{\eta} \ell^{-\frac{2r}{2r+\beta}}
     + 
    \frac{C_{\rm bias}}{c_{\rm obs}^{2}(h)} \mathbb E[H^{2p+2}],
\end{align}
which holds with probability at least $1-\eta$ over the model $\rho_{\mathcal X,\boldsymbol h}^{\rm design}$.
Fig.~\ref{fig:dep-graph} shows the road-map to prove the FS-HP bounds for the flow excess risk in Lemma~\ref{lm:flow-FS-HP-bound} and the FS-HP bounds for the vector field excess risk in Theorem~\ref{th:vec_field_pac}.

\begin{figure}[t]
    \centering
    \usetikzlibrary{matrix,positioning,arrows.meta,fit,backgrounds,calc}
    \begin{tikzpicture}[
      >=Latex,
      box/.style={
        draw, rounded corners, align=center, inner sep=3pt, 
        text width=32mm, 
        font=\small
        },
      ass/.style={box, dashed, fill=gray!10},
      lem/.style={box, fill=gray!15},
      thm/.style={box, fill=gray!25},
      cor/.style={box, fill=gray!10},
      prop/.style={box, fill=gray!10},
      group/.style={draw, dashed, rounded corners, inner sep=6pt, fill=gray!3}
    ]
    
    \matrix (M) [matrix of nodes, column sep=6mm, row sep=3mm]{
    |[ass] (A1)| {
        Assum.~\ref{as:bound_step_ratio}\\
        bounded step ratios/coeffs ($\zeta, \alpha, \beta$)
    }  & 
    |[ass] (SF)| {
        Filter qual./Lipschitz\\
        (Def.~\cite[Def.~1]{rosasco2005spectral})
    }\\
    |[ass] (A2)| {Assum.~\ref{as:vstep_consistency}\\order-$p$ consistency}      & 
    |[ass] (SCv)| {Source cond. on vector field\\$\Phi_\rho\in\Omega^{(\mathcal V)}_{r,R}$} \\
    |[ass] (A3)| {Assum.~\ref{as:ode_regularity}\\ODE regularity}               & 
    |[ass] (SCw)| {Source cond. on flow\\$\Phi_\rho\in\Omega^{(\mathcal W)}_{r,R}$} \\
    |[lem] (LTE)| {Lemma~\ref{lem:uniform-LTE-VM}\\Uniform LTE}                & 
    |[lem] (Obs)| {Lemma~\ref{lem:obs}\\Observability via $\mathsf B$} \\
    |[thm] (VF)| {Thm.~\ref{th:vec_field_pac}\\Vector-field FS--HP bound}
    & 
    |[lem] (Flow)| {Lemma~\ref{lm:flow-FS-HP-bound}\\Flow FS--HP bound} \\
    };
    
    \begin{pgfonlayer}{background}
      \node[group, fit=(A1)(A2)(A3),
            label={[yshift=0ex]above:\strut vLMM Assumptions}] (GvLMM) {};
      \node[group, fit=(SF)(SCw),
            label={[yshift=0ex]above:\strut {Spectral Regularization}}] (Gspectral) {};
    \end{pgfonlayer}
    
    \path[->]
      (A1.west) edge[bend right=24] (LTE.west)
      (A2.west) edge[bend right=16] (LTE.west)
      (A3.west) edge[bend right=12] (LTE.west)
    
      
      (SF.east) edge[bend left=12] (Flow.east)
      (SCw.east) edge[bend left=12] (Flow.east)
      (LTE.south) edge[bend right=0] (VF.north)
      (Obs.west) edge[bend right=12] (VF.east)
      (Flow.west) edge[bend left=12] (VF.east)

      (SCv) edge (SCw)
      ;
    
    
    
    \draw[->]
      (A1.east)
       .. controls +(6mm,0) and +(0mm,-0mm) ..
      ($ (A1)!0.5!(Obs) $)
       .. controls +(0mm,0mm) and +(-6mm,0) ..
      (Obs.west);
    
    \end{tikzpicture}
    \caption{Dependency graph linking assumptions, lemmas, and theorems
    used in the SFL analysis.}
    \label{fig:dep-graph}
\end{figure}

\section{Variable Step Linear Multistep Method (vLMM) Theory}\label{sec:variable}
This section records some vLMM facts we use to turn flow-space error into vector-field error within SFL. 
The role of vLMM in the SFL framework is twofold: 
(i) it supplies the \emph{windowed} surrogate that links vector fields to flows through the multistep forcing/aggregation operators (enabling observability), and 
(ii) it controls the \emph{deterministic discretization bias} via order-$p$ truncation guarantees.
Since, for variable-step methods, consistency depends only on step-ratio and not on absolute time shifts, a step-window is often parameterized by the step-ratio vector $\zeta$ and an anchor step $h_{\mathrm{anc}}$ in the vLMM framework. 

\subsection{General Form}
The \emph{vLMM general form} provides a linear map between an ODE and its solution on an irregular time grid $[h_j]_{-\infty}^\infty$.
Recall that a step-window of length $M$ on a time grid is defined as $T_M \coloneqq (h_{\rm anc}, \zeta(M-1))$.
For each solution $x_{k+M}$, we consider a step-window $T_M(k) = (h_{k+M-1}, \zeta_k(M-1))$, with $\zeta_k(M-1) \coloneqq [\zeta_k, \dots, \zeta_{k+M-1}]$.
\begin{definition}[vLMM General Form]\label{df:vlmm_univ_form}
    A vLMM takes
    the following general form~\cite[Eq.~5.15]{hairer1993solving}:
    \begin{equation}\label{eq:vlmm_unified}
        x_{k+M} + \sum_{j=1}^{M-1}\alpha_{jk}x_{k+j} = h_{k+M-1}\sum_{j=1}^M\beta_{jk}f(x_{k+j}),
    \end{equation}
    with coefficients: $\alpha_{jk} \coloneqq \alpha_{j}(\zeta_{k+1}, \dots, \zeta_{k+M-1}),
    \beta_{jk} \coloneqq\beta_{j}(\zeta_{k+1}, \dots, \zeta_{k+M-1})$,
    functions of the step ratios $\zeta_k(M-1)$.
\end{definition}

Before stating the deterministic bias assumptions, we instantiate the unified form \eqref{eq:vlmm_unified} with the standard variable-step families: Adams-Bashforth (AB), Adams-Moulton (AM), and backward-differentiation formulas (BDF) \cite{hairer1993solving}. 
Each of the above methods arises from a suitable choice of the ratio-dependent coefficients, $(\alpha_{jk}(\zeta),\beta_{jk}(\zeta))$, and admits a windowed-flow representation aligned with our SFL label operator, bringing explicit, implicit, and stiff-stable linear discretizations under one single framework.
Furthermore, notice that the popular Savitzky-Golay (SG) derivative estimators arise from local polynomial least-squares fitting on a sliding window \cite{savitzky1964smoothing,schafer2011savitzky}. 
Therefore, for uniform steps and one-sided windows, SG estimators reduce to the classical finite-difference derivative weights underlying multistep formulas. 
Within SFL framework, SG estimators can be seen as an alternative windowed label operator providing derivative labels complementary to vLMM flow labels.

\subsection{Assumptions for Deterministic Bias}
We now list the assumptions needed to obtain deterministic bounds on the vLMM-induced flow approximation error.
\begin{definition}[Order-$p$ consistency]~\cite[Def.~5.2]{hairer1993solving}\label{df:vstep_consistency}
    A vLMM \eqref{eq:vlmm_unified} is order-$p$ if:
    \begin{equation*}
        q(x_{k+M}) + \sum_{j=1}^{M-1}\alpha_{jk}q(x_{k+j}) = h_{k+M-1}\sum_{j=0}^M \beta_{jk} q'(x_{k+j})
    \end{equation*}
    holds for all polynomials $q(x)$ of degree $\leq p$ and grids $[h_j]_{-\infty}^\infty$.
\end{definition}
Notice that the window size $M$ is the number of back steps a $M$-step vLMM uses, whereas the \emph{order} $p$ is the accuracy of the method.
They are related but not identical: in classical families, AB has $p=M$, AM has $p=M+1$, and BDF has $p=M$, with zero-stability only up to $M\le 6$.
Moving onward, we state the key assumptions needed from the vLMM theory to establish an FS-HP bound in the SFL framework.
\begin{assumption}[Bounded step ratios and coefficients]\label{as:bound_step_ratio}
    There exist constants $0<\underline{\zeta}\le\bar{\zeta}<\infty$ such that:
  \[
  \underline{\zeta}\ \le\ \zeta_k\ \le\ \bar{\zeta},\qquad \forall k.
  \]
  Moreover, the variable-step coefficients $\{\alpha_{jk},\beta_{jk}\}$ are uniformly bounded.
\end{assumption}
\noindent Notice that bounded coefficients $\alpha_{jk}, \beta_{jk}$ follow from bounded step ratios for AB/AM methods and BDF \cite[Lem.~5.3]{hairer1993solving}.
\begin{assumption}\label{as:vstep_consistency}
    The vLMM \eqref{eq:vlmm_unified} is order-$p$ consistent in the sense of Definition~\ref{df:vstep_consistency}.
\end{assumption}
\begin{assumption}[ODE regularity]\label{as:ode_regularity}
    The right-hand side $f$ is sufficiently smooth, i.e., $f\in \mathcal{C}^{p+1}$ and locally Lipschitz on the forward-invariant compact set $\mathcal{X}$.
\end{assumption}

\section{Spectral Regularization Learning Theory}\label{sec:spectral}
This section assembles the spectral-regularization ingredients used by SFL. 
We adopt the operator-theoretic viewpoint of \cite{rosasco2005spectral,de2005learning,bauer2007regularization}: population and sampling Hilbert spaces are linked by inclusion maps whose covariance operators are spectrally filtered to stabilize inversion. 
Concretely, we fix the vector-field space $\mathcal V$, the \emph{population flow space} $\mathcal{W}^\circ$, and the \emph{windowed flow space} $\mathcal W$ on which learning acts; 
define the inclusions $\mathsf I_{\mathcal V},\mathsf I_{\mathcal W}$ with covariance operators $\mathsf T_{\mathcal V},\mathsf T_{\mathcal W}$ (and integral operators $\mathsf L_{\mathcal V},\mathsf L_{\mathcal W}$); 
and introduce a spectral filter ${\mathsf g}_\lambda$ with qualification $\nu$ and Lipschitz exponent $\mu$. 
In SFL, the estimator is obtained by filtering $\mathsf T_{\mathcal W}$ (componentwise in the vvRKHS), yielding a flow predictor in $\mathcal H_{\mathcal W}$ that we subsequently \emph{lift} to the vector field using the multistep observability inequality and vLMM bias control on irregular grids.

\subsection{Population and Sampling Hilbert Spaces}
We first encounter the notion of population flow $\mathcal{W}$ and vector-field space $\mathcal{V}$ in the problem formulation.
In this section, we refine the two notions by formally introducing the three population Hilbert spaces $\mathcal{V}, \mathcal{W}, \mathcal{W}^\circ$ and a finite-dimensional sampling space $\mathcal{Y}$. 
The population vector-field space $\mathcal{V}$ is defined as:
$$\mathcal{V}\coloneqq L^2(\mathcal X,\rho_{\mathcal X}^{\rm design};\mathbb R^n),\quad \mathcal{H}_{\mathcal V}\subset \mathcal{V},$$
where $\mathcal{H}_{\mathcal V}$ is a vvRKHS equipped with an inner product, $\langle\cdot,\cdot\rangle_{\mathcal{H}_{\mathcal V}}$, induced by the measurable, bounded port-Hamiltonian kernel $k_{\rm phs}$~\cite{beckers2022gaussian} and densely embedded in $\mathcal{V}$ with the inclusion map $\mathsf I_{\mathcal V}:\mathcal H_{\mathcal V} \to \mathcal{V}$. 
Recall that $\rho_{\mathcal X,\boldsymbol h}^{\rm design} \coloneqq \rho_{\mathcal X}^{\mathrm{design}} \times \rho_{\boldsymbol h}$, the single-step flow space is:
\begin{equation*}
    \mathcal{W}^\circ\coloneqq L^2\big(\mathcal X\times\boldsymbol h,\rho_{\mathcal X,\boldsymbol h}^{\rm design};\mathbb R^n\big), \quad \mathcal{H}_{\mathcal{W}^\circ}\subset \mathcal{W}^\circ,
\end{equation*}
with the embedded vvRKHS $\mathcal H_{\mathcal{W}^\circ}$ and the inclusion map defined as $\mathsf I_{\mathcal{W}^\circ}:\mathcal H_{\mathcal{W}^\circ} \to \mathcal{W}^\circ$.
Notice that $\mathcal{W}^\circ$ is the space in which the one-step exact flow $\Phi^h(x; f_\rho)$ lives.
For the windowed flow space, 
we recall the $M$-steps window $T_M \in \mathcal{T}$ and its distribution $\rho_{\mathcal{T}}$.
Denote the backshift times by $\tau_j(\zeta)\ge 0$ measured from $h_{k+M-1}$ and a composite measure $\rho^{\rm design}_{\mathcal X, \mathcal{T}} \coloneqq
\rho_{\mathcal X}^{\rm design}\times\rho_{\mathcal{T}}$. 
Define the windowed flow space:
\[
\mathcal{W}
\coloneqq L^2 \big(\mathcal{X} \times\mathcal{T},\rho^{\rm design}_{\mathcal X, \mathcal{T}};\ \mathbb R^n\big),\quad \mathcal H_{\mathcal{W}}\subset\mathcal{W}.
\]
The Hilbert space $\mathcal{W}$ is the space in which spectral regularization acts directly.
\begin{remark}
    Notice that $\mathcal{W}^\circ$ can be embedded in $\mathcal{W}$ using the isometric lifting operator $\mathsf{J}:\mathcal{W}^\circ\to{\mathcal{W}}$.
    For every one-step flow $\Phi \in \mathcal{W}^\circ$, there is a $\mathsf{J}\Phi \in \mathcal W$:
    \[
    (\mathsf{J}\Phi)(x;h, \zeta)\coloneqq\Phi^h(x),
    \]
    that is, $\mathsf{J} \Phi$ is independent of $\zeta$ ($\zeta$-constant). 
    An exact flow $\Phi^{h}(x;f)\in\mathcal{W}^\circ$ induced from $f$ is included in $\mathcal{W}$ via:
    \[
    \Phi^{h}(x; f) = (\mathsf{J}\Phi)(x;h, \zeta, f) \eqqcolon \Psi(x,T; f) \in \mathcal{W}.
    \]
\end{remark}
Finally, the sampling space for finite data is $\mathcal{Y}\subset \mathbb R^{n\ell}$, equipped with the empirical norm $\|y\|_{\mathcal{Y}}^2=\frac{1}{\ell}\sum_{i=1}^\ell\|y_i\|^2$, with $\ell$ denoting the number of effective data points.

\subsection{Inclusion Maps and Other Important Operators}
Denote the inclusion $\mathsf I_{\mathcal V}:\mathcal H_{\mathcal V}\to \mathcal{V}$.
Define the \emph{covariance operator} $\mathsf T_{\mathcal V}$ and \emph{integral operator} $\mathsf{L}_{\mathcal V}$:
\begin{equation}\label{eq:vf-cov-int-operate}
    \mathsf T_{\mathcal V} \coloneqq \mathsf I_{\mathcal V}^* \mathsf I_{\mathcal V}:\mathcal H_{\mathcal V}\to\mathcal H_{\mathcal V},\qquad \mathsf{L}_{\mathcal V} \coloneqq \mathsf I_{\mathcal V} \mathsf I_{\mathcal V}^*:\mathcal{V}\to \mathcal{V}.
\end{equation}
For learning in flow space we fix a single hypothesis vvRKHS $\mathcal H_{\mathcal W}\subset \mathcal{W}$ with bounded evaluation and inclusion $\mathsf I_{\mathcal W}:\mathcal H_{\mathcal W}\to \mathcal{W}$. 
Define:
\begin{equation}
    \label{eq:flow-cov-int-operate}
    \mathsf T_{\mathcal W} \coloneqq \mathsf I_{\mathcal W}^* \mathsf I_{\mathcal W}:\mathcal H_{\mathcal W}\to\mathcal H_{\mathcal W},\qquad \mathsf{L}_{\mathcal W} \coloneqq \mathsf I_{\mathcal W} \mathsf I_{\mathcal W}^*:\mathcal{W}\to \mathcal{W}.
\end{equation}
In spectral regularization learning theory, we view the finite dataset as inducing empirical analogues of these operators. Given a sample \(\mathcal D=\{(z_i,y_i)\}_{i=1}^\ell\) with windowed inputs \(z_i := (x_i,T_M(i))\) and outputs \(y_i\in\mathbb R^n\), let \(\mathsf S_z:\mathcal H_{\mathcal W}\to\mathcal Y\) be the sampling operator \((\mathsf S_z \Phi)_i \coloneqq \Phi(z_i)\) with empirical norm \(\|v\|_{\mathcal Y}^2\). 
The empirical covariance and empirical integral operators are:
\[
\widehat{\mathsf T}_{\mathcal W} \coloneqq \mathsf S_z^* \mathsf S_z:\mathcal H_{\mathcal W}\to\mathcal H_{\mathcal W},
\qquad
\widehat{\mathsf L}_{\mathcal W} \coloneqq \mathsf S_z \mathsf S_z^*:\mathcal Y\to\mathcal Y.
\]
In coordinates, \(\widehat{\mathsf T}_{\mathcal W}\) corresponds to the normalized information matrix built from feature vectors, while \(\widehat{\mathsf L}_{\mathcal W}\) reduces to the kernel Gram matrix acting on sample outputs. 
Spectral regularization theory views \(\widehat{\mathsf T}_{\mathcal W}\) and \(\widehat{\mathsf L}_{\mathcal W}\) as finite-rank perturbations of \(\mathsf T_{\mathcal W}\) and \(\mathsf L_{\mathcal W}\), so that learning becomes an inverse problem of recovering a target function through nearly singular operators. 
The spectral filter \(\mathsf g_\lambda\) stabilizes this inversion by damping directions associated with small eigenvalues.
We consider estimators $\Phi^{\lambda_\ell}_{z,y}\in\mathcal H_{\mathcal W}$ with spectral filter ${\mathsf g}_\lambda$ applied to $\mathsf T_{\mathcal W}$.

\subsection{Filter Qualification and Lipschitz Condition}
Another component we need for constructing the FS-HP bound is the notion of a spectral filter, its qualification, and the Lipschitz condition. 
\begin{definition}\cite[Def.~1]{rosasco2005spectral}
    Let ${\mathsf g}_\lambda: [0, \kappa^2]\to \mathbb{R}$ be a \emph{spectral filter}. 
    Its \emph{qualification} $\nu>0$ is the largest order such that:
    \[
    \sup_{0<\sigma\le \kappa^2} \big|1-\sigma {\mathsf g}_\lambda(\sigma)\big| \sigma^{\nu} \ \le\ \gamma_\nu \lambda^\nu,\qquad \forall 0<\nu\le \bar\nu.
    \]
    We also use the \emph{Lipschitz condition} $|{\mathsf g}_\lambda(\sigma)-{\mathsf g}_\lambda(\sigma')|\le L \lambda^{-\mu}|\sigma-\sigma'|$ for some $\mu>0$, and set $\beta=\max\{1,2\mu\}$.
\end{definition}
A spectral filter is an {approximate inverse} applied to the spectrum of a positive self-adjoint operator (e.g., the covariance/integral operator) via functional calculus: if a covariance operator $\mathsf T=\sum_i \sigma_i \langle\cdot,e_i\rangle e_i$, then ${\mathsf g}_\lambda(\mathsf T)=\sum_i {\mathsf g}_\lambda(\sigma_i)\langle\cdot, e_i\rangle e_i$. 
It attenuates directions with small eigenvalues $\sigma_i$ to stabilize the inversion, while the {residual} $1-\sigma {\mathsf g}_\lambda(\sigma)$ quantifies the bias introduced by regularization. 
The {qualification} $\nu$ is the largest order for which $r_\lambda(\sigma)\leq \lambda^\nu \sigma^{-\nu}$ uniformly. 
Typical values of $\nu$ are: Tikhonov with $\nu=1$; order-$t$ iterated Tikhonov with $\nu=t$; Landweber with $\nu \to \infty$~\cite[Sec.~3.1]{rosasco2005spectral}. 
The Lipschitz condition ensures stability of ${\mathsf g}_\lambda(\mathsf T)$ to empirical perturbations.

\subsection{Source Condition}


Target-kernel alignment is captured by an H\"older-type source condition \cite{engl2015regularization} with exponent \(r>0\); larger \(r\) yields faster rates up to the filter qualification \(\nu\). 
Conceptually, it measures how well \((\Phi_\rho,f_\rho)\) align with the vvRKHS prior. 
Formally, it requires membership in the range of a positive power of the integral operator \(\mathsf L=\mathsf I\mathsf I^*\). 
Since SFL learns in a windowed flow space and then lifts guarantees to the vector field, we impose sources at both levels: a flow source driving the flow FS-HP bound and a vector-field source encoding physics, e.g., Hamiltonian passivity, incompressibility \(\nabla\cdot f=0\), symmetry/equivariance. 
In practice, such structures can be enforced via physics-informed vvRKHS kernels, projection of the lifted estimator, or penalty terms \cite{beckers2022gaussian,baddoo2023physics,cai2021physics}.
For $r>0$ and $R_{\mathcal V}>0$, assume $f_{\rho}$ is the target vector field with the following structure defined through \eqref{eq:vf-cov-int-operate}:
\[
f_\rho \in \Omega^{(\mathcal{V})}_{r,R_{\mathcal V}}\ \coloneqq\ \big\{  f\in \mathcal{V} : f=\mathsf{L}_{\mathcal V}^r v,\ \|v\|_{\mathcal V}\le R_{\mathcal V}  \big\}.
\]
Likewise, for $R_{\mathcal W}>0$,
\[
\Phi_\rho \in \Omega^{(\mathcal{W})}_{r,R_{\mathcal W}}\ \coloneqq\ \big\{ \Phi\in \mathcal{W}:\ \Phi = \mathsf{L}_{\mathcal W}^r w,\ \|w\|_{\mathcal W}\le R_{\mathcal W}  \big\},
\]
where $\mathsf{L}_{\mathcal W}$ is defined in \eqref{eq:flow-cov-int-operate}.
This flow-level source is the one used in Lemma~\ref{lm:flow-FS-HP-bound} to obtain the FS-HP rate $\ell^{-\frac{2r}{2r+\beta}}$ with $\beta=\max\{1,2\mu\}$.

\section{Spectral Flow Learning Theory}\label{sec:sfl}
This section develops SFL's core machinery: variable-step windows yield well-posed labels via Koopman-lag forcing and state aggregation.
We define the label map and vLMM residual, prove an observability inequality and a uniform local-truncation bound, and derive FS-HP excess-risk bounds for flows and for vector fields.

\subsection{Label Map and vLMM Residual}
A key aspect of a theory on learning a vector field from trajectory data is the construction of well-defined input-label pairs.
In our framework, we leverage the Koopman-lag forcing operator $\mathsf{B}$ and state-aggregation operator $\mathsf{A}$ to define the population-level loss via the label map and vLMM residual.
Consider a vLMM as defined in Definition~\ref{df:vlmm_univ_form}. 
Fix the window and anchor it at $t_{k+M}$. 
Define the backward lags:
$$
\tau_0:=0,\qquad 
\tau_i:=\sum_{q=1}^{i} h_{k+M-q}\quad(i=1,\dots,M),
$$
so that each state in the window satisfies:
$$
x_{k+M-i}=\Phi^{-\tau_i} \big(x_{k+M}\big),\qquad i=0,\dots,M.
$$
Equivalently, for $j=1,\dots,M$,
$$
x_{k+j}=\Phi^{-\tau_{M-j}}\big(x_{k+M}\big).
$$
Using the Koopman action $(\mathcal{K}^{t}f)(x)=f(\Phi^{t}(x))$, the right-hand side of \eqref{eq:vlmm_unified} can be rewritten as:
$$
h_{k+M-1}\sum_{j=1}^{M}\beta_{jk}  (\mathcal{K}^{-\tau_{M-j}}f)\big(x_{k+M}\big).
$$
Renaming the indices as $\beta_j(\zeta)\coloneqq \beta_{M-j,k}$ and setting $x\coloneqq x_{k+M}$, we define the Koopman-lag forcing operator:
\begin{equation}\label{eq:koopman-lag-forcing}
    (\mathsf B f)(x;h,\zeta)
    \coloneqq h\sum_{j=0}^{M}\beta_j(\zeta) (\mathcal{K}^{-\tau_j}f)(x),
    \quad\mathsf{B}:\mathcal V\to\mathcal W,
\end{equation}
where $\tau_j \coloneqq \tau_j(\zeta)$ are the cumulative lags determined by $\zeta$ and the anchor step $h\coloneqq h_{k+M-1}$.
Similarly, the state-aggregation operator is defined as:
\begin{equation*}
    (\mathsf A\Phi)(x;h,\zeta)
    \coloneqq \sum_{j=1}^{M}\alpha_j(\zeta) \Phi^{-\tau_j}(x),
    \quad\mathsf{A}:\mathcal{W}^\circ\to\mathcal W.
\end{equation*}
Now, we are ready to define the population-level label map and vLMM residual.
\begin{definition}[Label Map and vLMM residual]
    The population-label map and vLMM residual are:
    \begin{equation}\label{eq:label-map_vlmm-residual}
        \mathsf{Y}_\alpha(\Phi)\coloneqq \mathsf J\Phi+\mathsf A\Phi \in\mathcal W,\ 
        \Delta(f)\coloneqq \mathsf{Y}_\alpha(\Phi)-\mathsf B f\ \in\mathcal W,
    \end{equation}
    respectively.
\end{definition}

\subsection{Observability Inequality and Local Truncation Error (LTE)}
Some useful consequences we can leverage from the above construction include the following.
\begin{lemma}[Observability via $\mathsf{B}$]\label{lem:obs}
Let: 
$$c_{\rm obs}(h)\coloneqq\inf_{\|v\|_{\mathcal{V}}=1}\|\mathsf{B}v\|_{\mathcal{W}}=\sqrt{\lambda_{\min}(\mathsf{B}^*\mathsf{B})};$$
then:
\[
\|\mathsf{B}(f-f_\rho)\|_{\mathcal W}^2 \geq c_{\rm obs}(h)^2\|f-f_\rho\|_{\mathcal V}^2,\qquad \forall f\in \mathcal{V}.
\]
\end{lemma}
\begin{proof}
    This is immediate from $\|Av\|^2\ge\lambda_{\min}(A^*A)\|v\|^2$ for any bounded operator $A$.
\end{proof}

\begin{lemma}[Uniform LTE for vLMM residual]\label{lem:uniform-LTE-VM}
Under Assumptions~\ref{as:bound_step_ratio}, \ref{as:vstep_consistency}, and \ref{as:ode_regularity},
there exists a constant
$C_{\mathrm{LTE}}>0$, independent of $t$ and $T$, such that the single-window residual is:
\begin{equation*}
    \|\Delta(f)(x(t);T)\|_2
    \le C_{\mathrm{LTE}}h^{p+1},
\end{equation*}
for all $t \geq 0$, $T\in\mathcal{T}$.
Consequently, in the flow space $L^2$ norm,
\begin{equation}\label{eq:LTE-flow-norm}
    \|\Delta(f)\|_{\mathcal W}
        \le\ C_{\mathrm{LTE}} \sqrt{\mathbb E[H^{2p+2}]}.
\end{equation}
\end{lemma}

\begin{proof}[Proof]
Let $x(\cdot)$ solve $\dot x=f(x(t)) + g(x(t)) u(t)$ on a forward‑invariant compact set $\mathcal{X}$ so that $\sup_t\|x^{(p+1)}(t)\|<\infty$, where $x^{(p+1)}$ denotes the $(p{+}1)$-order derivate of $x$. Following the notation in \cite{hairer1993solving}, we write the residual componentwise as the variable‑step difference operator:
\begin{equation*}
    \begin{aligned}
        &\mathcal L(x; T)
        = x(t{+}h) \\
        &\quad +\sum_{j=1}^{M-1}\alpha_j(\zeta) x(t{-}\tau_j(\zeta))
         -h\sum_{j=0}^{M}\beta_j(\zeta) \dot x(t{-}\tau_j(\zeta)).
    \end{aligned}
\end{equation*}
By Assumption~\ref{as:vstep_consistency}, which provides order-$p$ consistency, 
Assumption~\ref{as:bound_step_ratio} that guarantees bounded variable‑step coefficients, and 
Assumption~\ref{as:ode_regularity} for a smooth ODE, \cite[Eq.~5.17]{hairer1993solving} yields:
$$
\mathcal L(x; T) = \mathcal{O}(h^{p+1}),
$$
uniformly over admissible $T$ for sufficiently smooth $x(\cdot)$. 
A trivial extension of the constant‑step local truncation error lemma~\cite[Lem.~2.2]{hairer1993solving} to variable steps then implies that the local error at each anchor point is $\mathcal{O}(h^{p+1})$. 
Identifying
$\Delta(f)=\mathcal L(x; T)$ gives the pointwise bound with a constant $C_{\mathrm{LTE}}$ independent of $(t,T)$. 
This proves the first claim.

For the $\mathcal W$ norm, square and integrate over $\rho^{\rm design}_{\mathcal X,\mathcal{T}}$ to get:
\begin{align*}
    \|\Delta(f)\|_{\mathcal W}^{2}
    &=\int \|\Delta(f)(x;T)\|^{2}_2\ d \rho^{\rm design}_{\mathcal{X}, \mathcal{T}}\\
    &\le\ C_{\mathrm{LTE}}^{2}\int h^{2p+2} d\rho_{\boldsymbol h} = C_{\mathrm{LTE}}^{2} \mathbb E[H^{2p+2}].
\end{align*}
Taking the square root of both sides yields \eqref{eq:LTE-flow-norm}.
\end{proof}

Lemmas~\ref{lem:obs} and \ref{lem:uniform-LTE-VM} provide stepping stones needed to lift flow-space guarantees to the vector field: a $\mathsf{B}$-observability inequality that controls $\|f-f_\rho\|_{\mathcal V}$ by $\|\mathsf B(f-f_\rho)\|_{\mathcal W}$, and a uniform bound of the vLMM modeling residual via the moment $\mathbb E[H^{2p+2}]$. 
Next, we formalize the flow regression problem, introduce the excess risk $\mathcal E(\Phi)$ and the target regression function $\Phi_\rho$, and state the flow FS-HP bound obtained by spectral regularization.
Combining that bound with the bounded label operator $\mathsf{Y}_\alpha=\mathsf J+\mathsf A$ and the decomposition
$\mathsf B(\widehat f-f_\rho)=\Delta(\widehat f)+{\mathsf Y}_\alpha(\widehat\Phi-\Phi_\rho)+\Delta(f_\rho)$
yields a two-term vector-field bound featuring a statistical rate in $\ell$ and a deterministic vLMM bias in $\mathbb E[H^{2p+2}]$, scaled by $c_{\rm obs}(h)$. 

\subsection{FS-HP Excess Risk Bound for Flow-Learning}  \label{sec:FL-bound}
We begin by introducing the definition of excess risk, following the formality presented in~\cite{rosasco2005spectral}. 
\begin{definition}[Excess Risk]
    The excess risk of a predictor $\Pi$ with respect to its target $\Pi_\rho$, i.e., the best possible predictor under a design law $\rho$, is $\mathcal{E}(\Pi) \coloneqq \mathcal{R}(\Pi) - \mathcal{R}(\Pi_\rho)$, with the expected risk of $\Pi$ and target predictor $\Pi_\rho$ given as:
    \begin{equation*}
        \mathcal{R}(\Pi)
        \coloneqq\int_{\mathcal{X}\times \mathcal{Y}}\|\Pi(x)-y\|^2 d\rho(x,y),\ \Pi_\rho(\cdot) \coloneqq  \int_{\mathcal Y} y d\rho(y\mid \cdot),
    \end{equation*}
    respectively.
\end{definition}
With flow label $Y\in\mathcal X$ and $\rho((x,T),y)=\rho(y\mid x,h)\rho_{\mathcal X,\mathcal{T}}^{\rm design}(x,T)$, the target flow predictor is:
\begin{align*}
    \Phi^h_\rho(x) &\coloneqq \mathbb E[Y\mid X=x,H=h]\\
    &=\int_{\mathcal X}y d\rho(y\mid x,h).
\end{align*}
For any flow predictor $\Phi\in \mathcal{W}$, let $\mathcal{Z}=\mathcal{X}\times \mathcal{T}$:
\begin{equation}\label{eq:excess-error}
\begin{aligned}
    \mathcal{R}(\Phi)
    &\coloneqq\int_{\mathcal{Z}\times \mathcal{X}}\|\Phi^h(x)-y\|^2 d\rho((x,T),y)\\
    &=\|\Phi-\Phi_\rho\|_{\mathcal{W}}^{2}+ \mathcal{R}(\Phi_\rho).
\end{aligned}
\end{equation}
Therefore, minimizing $\mathcal{E}(\Pi)$ is equivalent to 
approximating $\Pi_\rho$ with $\Phi$ via minimizing $\|\Pi - \Pi_\rho\|$~\cite[Def.~1]{rosasco2005spectral}.
The choice of the estimated flow $\widehat{\Phi} \coloneqq \Phi_{\mathcal{D}}^{\lambda_\ell}$ is inferred from data $\mathcal{D} = {\{z_i, y_i\}}_{i=1}^\ell$ and induced from inductive bias encoded in the source condition.
Under the state design law $\rho_{\mathcal X}^{\mathrm{design}}$, flow
excess risk is probabilistically bounded as follows.

\begin{lemma}[Flow FS-HP bound]\label{lm:flow-FS-HP-bound}
    Suppose $\Phi_\rho \in \Omega^{(\mathcal{W})}_{r, R}$ and $r \leq \nu$, which is the qualification $\mathsf g_{\lambda}$.
    For $0 < \eta \leq 1$, the following inequality holds with probability at least $1-\eta$:
    \begin{equation}
        \|\Phi - \Phi_\rho\|_{\mathcal W}
        \le 
        \log\frac{4}{\eta} C_{\rm flow} \ell^{-\frac{2r}{2r+\beta}},
    \end{equation}
    for some finite positive constant $C_{\rm flow}$ and with $\beta = \max\{1, 2\mu\}$.
\end{lemma}
\begin{proof}
    \cite[Thm.~2]{rosasco2005spectral} applies unchanged to $\mathcal H_{\mathcal W}\subset \mathcal{W}$.
\end{proof}

The rate in Lemma~\ref{lm:flow-FS-HP-bound} isolates the sample-size $\ell$ dependence from other problem-specific quantities. 
All other problem- and algorithm-specific quantities are collected into:
\begin{equation*}
    C_{\rm flow} = 2C_1^2+2\gamma_r^2R^2,
\end{equation*}
with $C_1  =  4\sqrt{2} \kappa M(\sqrt{DB}+\kappa^{5/2}L)$,
exactly as in the finite-sample bound in~\cite[Thm.~2]{rosasco2005spectral}. 
Here $\kappa^2=\sup_x K(x,x)$ bounds the kernel, $M=\sup|x|$ bounds the labels, while $B,D$ and $(L,\mu)$ are the spectral-filter constants defined by
$\sup_{0<\sigma\le\kappa^2}|g_\lambda(\sigma)|\le B/\lambda$, $\sup_{0<\sigma\le\kappa^2}|\sigma g_\lambda(\sigma)|\le D$, and
$|g_\lambda(\sigma)-g_\lambda(\sigma')|\le L \lambda^{-\mu}|\sigma-\sigma'|$, with $\beta\coloneqq\max\{1,2\mu\}$. The constant $\gamma_r$ is the qualification residual at order $r$:
\[
\sup_{0<\sigma\le\kappa^2}\bigl|1-\sigma g_\lambda(\sigma)\bigr| \sigma^{r} \le \gamma_r \lambda^{r},
\]
which yields the order-$\lambda^{2r}$ approximation term when $\Phi_\rho\in\Omega^{(\mathcal W)}_{r,R}$.
All these ingredients combine to give:
\[
\|\widehat\Phi-\Phi_\rho\|_{\mathcal W}^2
 \le 
\log \frac{4}{\eta} \bigl(2C_1^2+2\gamma_r^2R^2\bigr) \ell^{-\frac{2r}{2r+\beta}}.
\]
The factor $\log(4/\eta)$ arises from concentration bounds for Hilbert-valued averages.
The concentration bounds follow from the Pinelis-Sakhanenko inequality for Hilbert spaces~\cite[Lemma~8]{de2005learning} and can be proved using~\cite[Theorem~3.3.4]{yurinsky2006sums}.

The flow FS-HP bound 
is particularly useful when we assume that the target dynamics evolves discretely with a deterministic uniform step schedule, where the step size is irrelevant to the target flow \(\Phi_\rho\). 
In SFL, methods targeting such \(\Phi_\rho\) differ mainly in the hypothesis class \(\mathcal H_{\mathcal W}\), 
the source condition \(\Phi_\rho \in \Omega^{(\mathcal W)}_{r,R_{\mathcal W}} = \{\Phi = \mathsf L_{\mathcal W}^r w : \|w\|_{\mathcal W} \le R_{\mathcal W}\}\), spectral regularizer \(\mathsf g_\lambda\), and qualification \(\nu\).
Pipelines are therefore characterized by their choice of \(\Omega^{(\mathcal W)}_{r,R_{\mathcal W}}\) and \(\mathsf g_\lambda\), which together govern the rate and stability of the excess risk.
\begin{example}[Methods with Linear/Dictionary Feature Map]
    Let $\mathcal H_{\mathcal W}$ be the vvRKHS induced by a finite dictionary
    $\psi(x)\in\mathbb R^n$ with kernel $K(x,x')=\psi(x)^\top\psi(x')$. 
    The population integral operator
    $\mathsf L_{\mathcal W}$ acts on this finite-dimensional space. 
    If the model is well specified
    ($\Phi_\rho\in\operatorname{span}\{\psi_j\}$), then $\Phi_\rho=\mathsf L_{\mathcal W}^r w$ holds for any $r\ge0$
    with $R_{\mathcal W}$ determined by the coefficient vector. 
    If misspecified, $r>0$ measures alignment
    of $\Phi_\rho$ with the principal directions of $\mathsf L_{\mathcal W}$.
    System identification methods that fall into this category include: ARX and DMDc~\cite{proctor2016dynamic}.
\end{example}
On the other hand, Bayesian methods such as Gaussian process state space models (GP-SSM) are also naturally covered by the SFL framework, with a side note that the regularization coefficient $\lambda$ is chosen to be $\sigma^2_y / \ell$, where $\sigma_y$ is the standard deviation of the additive observation noise.

\begin{remark}[TSVD / spectral cut-off]
    Spectral cut-off (TSVD), \(\mathsf g_{\lambda}^{\mathrm{cut}}(\sigma)=\sigma^{-1}\mathbf{1}\{\sigma\ge\lambda\}\), has arbitrarily high qualification but violates the Lipschitz condition used in standard finite-sample proofs \cite{rosasco2005spectral,de2005learning}, so the baseline FS-HP theory does not apply verbatim. 
    Nonetheless, TSVD can be handled within extended spectral-regularization frameworks for non-smooth filters and projection schemes, yielding comparable rates under suitable assumptions; see \cite{engl2015regularization,bauer2007regularization}.
\end{remark}
Table~\ref{tab:flow-learners} provides a more detailed listing of various system identification methods with their corresponding hypothesis class and regularizer $\mathsf g_\lambda$.

\begin{table}[t]
    \centering
    \begin{tabular}{lcc}
        \toprule
        Pipeline & Hypothesis in $\mathcal H_{\mathcal W}$ & Filter $\mathsf g_\lambda$\\
        \midrule
        ARX/ARMAX/BJ & Linear/dictionary on $x, u$ & Ridge or TSVD\\
        NARX & Finite nonlinear dictionary & Ridge or TSVD\\
        Subspace ID & Past \& future Hankel stack & TSVD (order choice)\\
        DMD/DMDc & Linear in $x, u$ & Ridge or TSVD\\
        EDMD/HAVOK & Linear in $\psi(x)$ (delays) & Ridge or TSVD\\
        KDMD & Kernel ridge in $x$ & Tikhonov ($\nu{=}1$)\\
        GP-SSM & GP mean in RKHS & Tikhonov ($\lambda{=}\sigma_y^2/\ell$)\\
        \bottomrule
    \end{tabular}
    \caption{Classic flow-learning pipelines as spectral-filter estimators in the windowed flow space. Lemma~\ref{lm:flow-FS-HP-bound} applies verbatim.}
    \label{tab:flow-learners}
\end{table}

\subsection{FS-HP Excess Risk Bound for Vector-Field Learning}  \label{sec:Vbound}
Since 
discretization bias often matters in practice, we develop an FS-HP excess-risk bound for vector-field learning.
\begin{theorem}[Vector-field FS-HP bound]\label{th:vec_field_pac}
    Suppose $\Phi_\rho \in \Omega^{(\mathcal{W})}_{r, R}$ and $r \leq \nu$, which is the qualification $\mathsf g_{\lambda}$.
    For $0 < \eta \leq 1$, the following inequality holds with probability at least $1-\eta$:
    \begin{equation}\label{eq:vf-pac-bound}
        \|\widehat f - f_\rho\|_{\mathcal{V}}^{2}
        \le
        \frac{C_{\rm fit}}{c_{\rm obs}^{2}(h)} \log\frac{4}{\eta} \ell^{-\frac{2r}{2r+\beta}}
         + 
        \frac{C_{\rm bias}}{c_{\rm obs}^{2}(h)}\ \mathbb E[H^{2p+2}].
    \end{equation}
    for some finite positive constant $c_{\rm obs}(h)$, $C_{\rm fit}$, and $C_{\rm bias}$.
\end{theorem}
\begin{proof}
    Let $\widehat f \coloneqq f_{z,y}^{\lambda_\ell}$ be the vector‑field estimator.
    Then:
    \begin{align}
        \| \widehat f - f_\rho\|_{\mathcal{V}}^{2}
        &\le \frac{1}{c_{\rm obs}^{2}(h)}\ \|\mathsf{B}(\widehat f - f_\rho)\|_{\mathcal{W}}^{2} \label{eq:th:vec_field_pac:pf:B-obs}\\
        &\le \frac{1}{c_{\rm obs}^2(h)}(
        \|\underbrace{\mathsf{B}\widehat{f}-\mathsf{Y}_\alpha \widehat{\Phi}}_{\Delta(\widehat{f})}\|_{\mathcal W}^{2}\ +
        \|\mathsf{Y}_\alpha(\widehat{\Phi}-\Phi_\rho)\|_{\mathcal W}^{2}\nonumber\\
        &\qquad +\|\underbrace{\mathsf{Y}_\alpha \Phi_\rho -\mathsf{B}f_\rho}_{\Delta(f_\rho)}\|_{\mathcal W}^{2})\label{eq:th:vec_field_pac:pf:tri_ineq}\\
        &\le \frac{1}{c_{\rm obs}^{2}(h)}\ \|\mathsf{J} + \mathsf{A}\|^2_{\mathcal{L}(\mathcal{W})}
        \big(\log\frac{4}{\eta}C_{\rm flow}\ell^{-\frac{2r}{2r+\beta}}\big)\nonumber\\
        &\qquad + \frac{1}{c_{\rm obs}^{2}(h)}\big(2C^2_{{\rm LTE}} \big)\mathbb{E}[H^{2p+2}],\label{eq:th:vec_field_pac:pf:LTE_flow-fshp-bound}
    \end{align}
    with $\mathcal{L}(\mathcal W)$ the space of bounded linear operators $\mathcal{W}\to \mathcal{W}$.

    First, inequality \eqref{eq:th:vec_field_pac:pf:B-obs} follows immediately from the Lemma~\ref{lem:obs}. Indeed, by definition,
    $c_{\rm obs}^2(h)=\lambda_{\min}(\mathsf B^*\mathsf B)$, so for any $v\in\mathcal V$,
    $c_{\rm obs}^2(h)\|v\|_{\mathcal V}^2\le\|\mathsf B v\|_{\mathcal W}^2$; applying this with
    $v=\widehat f-f_\rho$ yields:
    $$
    \|\widehat f-f_\rho\|_{\mathcal V}^2
     \le  \frac{1}{c_{\rm obs}^2(h)} \|\mathsf B(\widehat f-f_\rho)\|_{\mathcal W}^2,
    $$
    which is exactly \eqref{eq:th:vec_field_pac:pf:B-obs}.
    
    Next, to obtain \eqref{eq:th:vec_field_pac:pf:tri_ineq}, we add and subtract
    $\mathsf Y_\alpha\widehat\Phi$ and $\mathsf Y_\alpha\Phi_\rho$ inside the norm, and then apply the triangle inequality in $\mathcal W$ together with the linearity of the label map
    $\mathsf Y_\alpha=\mathsf J+\mathsf A$ as defined in \eqref{eq:label-map_vlmm-residual}. 
    Using the identity:
    $$
    \mathsf B(\widehat f-f_\rho)
    =\underbrace{\mathsf B\widehat f-\mathsf Y_\alpha\widehat\Phi}_{\Delta(\widehat f)}
    +\underbrace{\mathsf Y_\alpha(\widehat\Phi-\Phi_\rho)}_{\text{flow mismatch}}
    +\underbrace{\mathsf Y_\alpha\Phi_\rho-\mathsf B f_\rho}_{\Delta(f_\rho)},
    $$
    we directly arrive at the sum of the three squared $\mathcal W$-norms displayed in
    \eqref{eq:th:vec_field_pac:pf:tri_ineq}.
    
    For the first term on the right-hand side of
    \eqref{eq:th:vec_field_pac:pf:LTE_flow-fshp-bound}, i.e., the contribution of the flow mismatch,
    we use that $\mathsf Y_\alpha=\mathsf J+\mathsf A$ is a bounded linear operator on $\mathcal W$.
    In fact, $\mathsf J$ is an isometry and $\mathsf A$ is bounded by Assumption~\ref{as:bound_step_ratio}.
    Hence, by sub-multiplicativity of the operator norm on $\mathcal W$,
    \begin{align*}
    \|\mathsf Y_\alpha(\widehat\Phi-\Phi_\rho)\|_{\mathcal W}^2
    &=\|(\mathsf J+\mathsf A)(\widehat\Phi-\Phi_\rho)\|_{\mathcal W}^2\\
     &\le  \|\mathsf J+\mathsf A\|_{\mathcal L(\mathcal W)}^2 \|\widehat\Phi-\Phi_\rho\|_{\mathcal W}^2.
    \end{align*}
    Invoking the flow FS-HP bound in Lemma~\ref{lm:flow-FS-HP-bound} for
    $\|\widehat\Phi-\Phi_\rho\|_{\mathcal W}^2$ then yields the first term appearing in
    \eqref{eq:th:vec_field_pac:pf:LTE_flow-fshp-bound}.
    
    For the second term on the right-hand side of
    \eqref{eq:th:vec_field_pac:pf:LTE_flow-fshp-bound}, i.e., the contribution of the two vLMM residuals,
    we apply Lemma~\ref{lem:uniform-LTE-VM} to the first and third terms of
    \eqref{eq:th:vec_field_pac:pf:tri_ineq}. 
    This gives
    $\|\Delta(\widehat f)\|_{\mathcal W}^2 \le C_{{\rm LTE}}^2 \mathbb E[H^{2p+2}],$
    which produce the second term in \eqref{eq:th:vec_field_pac:pf:LTE_flow-fshp-bound}.
    
    Finally, note that $\|\mathsf J+\mathsf A\|_{\mathcal L(\mathcal W)}$ is bounded under
    Assumption~\ref{as:bound_step_ratio}. Therefore, we can absorb constants by setting
    $C_{\rm fit}\ge C_{\rm flow} \|\mathsf J+\mathsf A\|_{\mathcal L(\mathcal W)}^2$ and
    $C_{\rm bias}\coloneqq 2C_{{\rm LTE}}^2$. Substituting these into
    \eqref{eq:th:vec_field_pac:pf:LTE_flow-fshp-bound} completes the proof of \eqref{eq:vf-pac-bound}.
\end{proof}

An insight from \eqref{eq:vf-pac-bound} is that the excess risk is composed of: 
(i) the statistical fitting error term $\log(4/\eta) \ell^{-{2r}/{(2r+\beta)}}$ that scales down as the number of sampling data increases; and
(ii) the discretization error term $\mathbb{E}[H^{2p+2}]$ that scales down as the average step size goes to zero.
The discretization error is independent of $\ell$, therefore, Theorem~\ref{th:vec_field_pac} predicts that we need to control both $\ell$ and $H$ to reduce the excess risk of $\widehat f$, and consequently $\widehat{\Phi}$.
Applications of Theorem~\ref{th:vec_field_pac} includes~\cite {leung2025learning}, and situations when the system identification algorithm explicitly depends on a linear discretization step.

\section{Numerical Simulations}\label{sec:numeric}
We consider a controlled mass--spring system
\begin{equation}
    \dot{x}(t)=
    \begin{bmatrix}
    0 & 1\\ -k/m & -c/m
    \end{bmatrix} x(t)+\begin{bmatrix}
    0\\ 1/m
    \end{bmatrix} u(t),
\label{eq:plant}
\end{equation}
with mass $m=2 \  \mathrm{kg}$, no viscous damping ($c=0$), and spring constant $k = m\omega_n^2$, where $\omega_n = 2\pi\cdot 0.4~\mathrm{rad/s}$. 
The initial condition is $x(0)=(0,0)^\top$ and $u(t)$ is a pseudo-random binary signal. 
Trajectories are integrated in $t \in [0,10]$ and evaluated with a uniform period $h = 0.004~\mathrm{s}$.
We evaluate identification performance both at a fixed sampling period and across a range of values.

\subsection{Benchmarking Across Source Conditions}

\begin{figure}[t]
    \centering
    \begin{subfigure}{0.49\linewidth}
        \includegraphics[width=\linewidth]{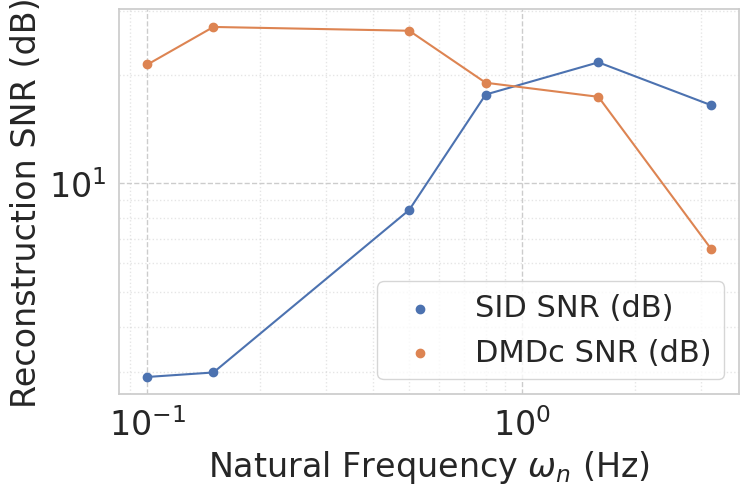}
    \end{subfigure}
    \begin{subfigure}{0.49\linewidth}
        \includegraphics[width=\linewidth]{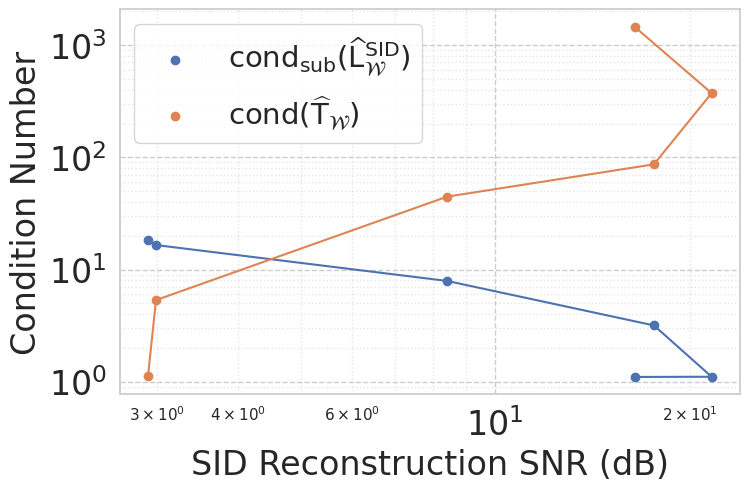}
    \end{subfigure}
    \caption{Reconstruction Fidelity: (Left) Compare reconstruction SNR of SID with DMDc with increasing natural frequency in \eqref{eq:plant}; (Right) Reconstruction SNR of SID against condition number of $\widehat{\mathsf{L}}_{\mathcal W}^{\rm SID}$ and $\widehat{\mathsf{T}}_{\mathcal W}$.}
    \label{fig:recons_snr_relation}
\end{figure}

The simulation in Fig.~\ref{fig:recons_snr_relation} shows that differing source conditions make algorithms excel on different systems.
Furthermore, benchmarking with a single-signal setup can obscure ill-conditioned regimes that basic diagnostics miss.
For a mass-spring system with fixed sampling rate and observation horizon, the dynamic mode decomposition with control (DMDc) reconstruction signal-to-noise ratio (SNR) decreases as the natural frequency decreases.
The varying reconstruction SNR of DMDc aligns with the decreasing condition number and increasing $\lambda_{\min}$ of the empirical covariance matrix $\widehat{\mathsf{T}}_{\mathcal W}$. 
By contrast, subspace identification (SID) exhibits the opposite behavior. 
The reason is that SID performs an orthogonal projection between the past and future block-Hankel matrices before the regression step, which effectively changes its source condition. 
If we account for this different source condition by computing a two-eigenvalue ``sub-condition number'' from the SID, specific integral operator $\widehat{\mathsf L}^{\mathrm{SID}}_{\mathcal W}$, then the expected downward trend is recovered: reconstruction SNR declines as the sub-condition number of $\widehat{\mathsf L}^{\mathrm{SID}}_{\mathcal W}$ increases in Fig.~\ref{fig:recons_snr_relation} (Right).


\subsection{The Fast Sampling Paradox}

\begin{figure}[t]
    \centering
    \begin{subfigure}{0.49\linewidth}
        \includegraphics[width=\linewidth]{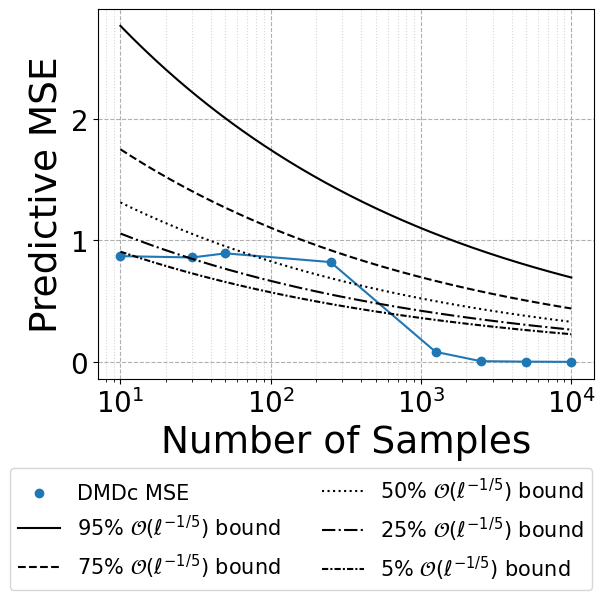}
    \end{subfigure}
    \begin{subfigure}{0.49\linewidth}
        \includegraphics[width=\linewidth]{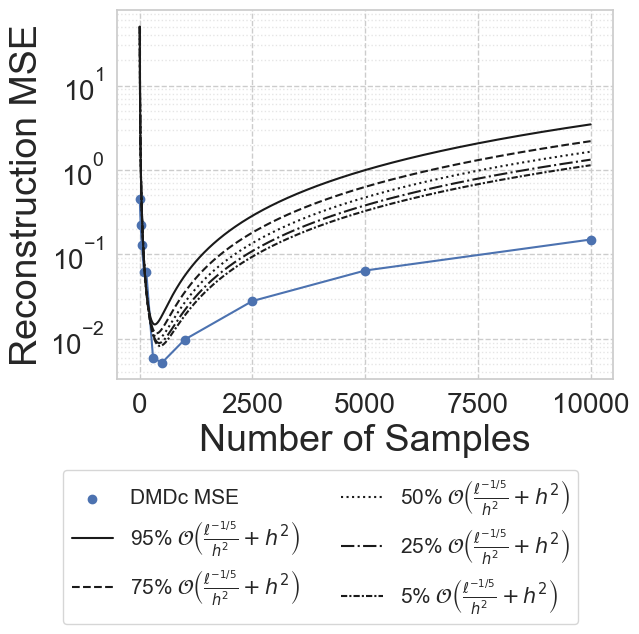}
    \end{subfigure}
    \caption{Fitting Performance: (Left) predictive mean square error (MSE) of DMDc prediction with increasing effective samples $\ell_{\rm eff}$; (Right) reconstruction MSE of DMDc with respect to increasing sample size drawn from the same horizon.}
    \label{fig:pred_perf}
\end{figure}

We now illustrate excess error bounds provided in Lemma~\ref{lm:flow-FS-HP-bound} and Theorem~\ref{th:vec_field_pac}.
In Fig.~\ref{fig:pred_perf} (Left), samples are collected in sampling horizons $(0.04, \dots, 40)$ with uniform step size $0.004~\mathrm{s}$, then the fitted DMDc MSE are evaluated on the time grid $(0, 0.004, \dots, 40)$.
Lemma~\ref{lm:flow-FS-HP-bound} predicts a monotonic decrease in excess risk as sample size increases, assuming step laws are held strictly fixed and discretization effects do not contribute to the resulting error. 
By contrast, Theorem~\ref{th:vec_field_pac} predicts a non-monotonic rebound when the identified object is a vector field and the discretization bias is amplified as step sizes diminish. 
In Fig.~\ref{fig:pred_perf} (Right), as the number of samples over a fixed observation horizon increases, the DMDc reconstruction MSE initially decreases. 
Beyond about 500 samples, however, the MSE rebounds, consistent with Theorem~\ref{th:vec_field_pac}.
With moderate noise, per-step SNR drops and fast-sampling artifacts emerge, biasing reconstruction unless explicitly handled by the method.

\section{Conclusion}
We proposed \emph{spectral flow learning}, which learns in a windowed flow space via a lag-linear label operator and spectral regularization, yielding FS-HP guarantees that lift to the vector field through a multistep observability inequality and a uniform local-truncation bound. 
The resulting two-term bound separates a statistical rate \(\mathcal{O}\left(\ell^{-\frac{2r}{2r+\beta}}\right)\) from a deterministic vLMM bias \(\mathbb{E}[H^{2p+2}]\).
The analysis spans AB/AM/BDF families and is compatible with physics-informed vvRKHS designs. 
SFL offers a common language for practical identification phenomena and motivates future work on principled choices of sampling rates and step schedules.

\bibliographystyle{IEEEtran}
\bibliography{refs}

\end{document}